%% file: preg_enhanced13.tex
\def\draft{0}

\documentclass{article}
\usepackage{amsfonts, amsmath, amssymb}
\usepackage{url, graphics}
\usepackage[dvips]{graphicx}
\usepackage{color}
\usepackage{fullpage}

\usepackage{amsthm}
\theoremstyle{plain}

\newtheorem*{theorem*}{Theorem}
\newtheorem{theorem}{Theorem}

\input{new_macros}

\numberwithin{theorem}{section}
\numberwithin{definition}{section}
\numberwithin{lemma}{section}
\numberwithin{proposition}{section}
\numberwithin{corollary}{section}
\numberwithin{claim}{section}
\numberwithin{fact}{section}

\newcommand{\altketbra}[1]{\ketbra{#1}{#1}}
\newcommand{\kb}[1]{\ketbra{#1}{#1}}

\newcommand{\F}{\operatorname{F}}
\newcommand{\Fbar}{\overline{F}}

\newcommand{\Bob}{\mbox{Bob}}
\newcommand{\Alice}{\mbox{Alice}}
\renewcommand{\X}{\widetilde{X}}
\renewcommand{\Y}{\widetilde{Y}}

\title{Parallel Repetition of Free Entangled Games: \\ Simplification and Improvements}
\author{Andr\'e Chailloux \thanks{SECRET Project Team, INRIA Paris-Rocquencourt.} \and
Giannicola Scarpa \thanks{Universitat Aut\`onoma de Barcelona, supported by EC project RAQUEL (323 970).}}

\begin{document}
\maketitle

\begin{abstract}
In a two-player game, two cooperating but non communicating players, Alice and Bob, receive inputs taken from a probability distribution. Each of them produces an output and they win the game if they satisfy some predicate on their inputs/outputs. The entangled value $\omega^*(G)$ of a game $G$ is the maximum probability that Alice and Bob can win the game if they are allowed to share an entangled state prior to receiving their inputs. 

The $n$-fold parallel repetition $G^n$ of $G$ consists of $n$ instances of $G$ where Alice and Bob receive all the inputs at the same time and must produce all the outputs at the same time. They win $G^n$ if they win each instance of $G$. Recently, there has been a series of works showing parallel repetition with exponential decay for projection games \cite{DSV13}, games on the uniform distribution \cite{CS14} and for free games, \emph{i.e.}, games on a product distribution \cite{JPY13}. 

This article is meant to be a follow up of \cite{CS14}, where we improve and simplify several parts of our previous paper. Our main result is that for any free game $G$ with value $\omega^*(G)=1-\eps$, we have $\omega^*(G^n) \le (1 - \eps^2)^{\Omega(\frac{n}{\log(l)})}$ where $l$ is the size of the output set of the game. This result improves on both the results in \cite{JPY13} and \cite{CS14}. The framework we use can also be extended to free projection games. We show that for a free projection game $G$ with value $\omega^*(G)=1-\eps$, we have $\omega^*(G^n) \le (1 - \eps)^{\Omega(n)}$.
\end{abstract}

\section{Introduction}
A \emph{two-player (nonlocal) game} is played between two cooperating parties, Alice and Bob, which are not allowed to communicate.  This game $G$ is characterized by an input set $I$, an output set $O$, a probability distribution $p$ on $I^2$
and a result function $V: O^2 \times I^2 \rightarrow \zo$. The game proceeds as follows: Alice receives $x \in I$, Bob receives $y \in I$ where $(x,y)$ is taken according to some distribution $p$. Alice outputs $a \in O$ and Bob outputs $b \in O$. They win the game if $V(a,b|x,y) = 1$. The value of the game $\omega(G)$ is the maximum probability, over all classical strategies, with which Alice and Bob can win the game. 

The $n$-fold parallel repetition $G^n$ of $G$ consists of the following. Alice and Bob get inputs $x_1,\dots,x_n$ and $y_1,\dots,y_n$, respectively. Each $(x_i,y_i)$ is taken according to $p$. They output $a_1,\dots,a_n$ and $b_1,\dots,b_n$, respectively. They win the game if and only if $\forall i, \ V(a_i,b_i|x_i,y_i) = 1$. In order to win the $n$-fold repetition, Alice and Bob can just take the best strategy for $G$ and use it $n$ times. If they do so, they will win $G^n$ with probability $(\omega(G))^n$ which shows that $\omega(G^n) \geq (\omega(G))^n$.

Parallel repetition of games studies how the quantity $\omega(G^n)$ behaves. For example, if $\omega(G^n) = (\omega(G))^n$ for each $n$ then we say that $G$ admits perfect parallel repetition. However, we know some games for which this does not hold. It was a long-standing open question to determine whether the value of $\omega(G^n)$ decreases exponentially in $n$. This was first shown by Raz~\cite{Raz98}. Afterwards, a series of works showed improved results for specific types of games~\cite{Hol07,Rao08,AKK+08,Raz11,BG14}. Parallel repetition for games has many applications, from direct product theorems in communication complexity~\cite{PRW97} to hardness of approximation results~\cite{BGS98,Fei98,Has01}.

In the quantum setting, it is natural to consider games where Alice and Bob are allowed to share some entangled state at the beginning of the game, before the inputs are generated. Entangled games exhibit Bell violations which are a witness of quantum non-locality. The study of entangled games may also be related to several aspects of quantum complexity, as in the classical setting.

Perfect parallel repetition has been shown for entangled XOR games~\cite{CSUU08}. It was also shown that entangled unique games~\cite{KRT08} admit parallel repetition with exponential decay. Finally, it was shown that any entangled game admits parallel repetition~\cite{KV11}. However, this last parallel repetition only shows a polynomial decay of $\omega^*(G^n)$. It was unknown for a large class of games whether this decay is exponential or not.

Recently, parallel repetition result with exponential decay has been shown for entangled projection games~\cite{DSV13} (see Section \ref{Section:EntangledGames} for a definition of projection games). We have also presented earlier a parallel repetition result with exponential decay for games on the uniform distribution. (Note that here and in the rest of the paper, unless otherwise stated, we use the convention that $\eps = 1-\omega^*(G)$.)
\begin{theorem*}[\cite{CS14}]
For any game $G$ on the uniform distribution, we have $\omega^*(G^n) \le (1 - \eps^2)^{\Omega(\frac{n}{\log(k) + \log(l)})}$ where $k$ and $l$ are respectively the dimension of the input set and of the output set.
\end{theorem*}
Independently Jain \etal presented a parallel repetition result with exponential decay on free games, which are games on a product distribution.
\begin{theorem*}[\cite{JPY13}]
For any game $G$ on a product distribution, we have $\omega^*(G^n) \le (1 - \eps^3)^{\Omega(\frac{n}{\log(l)})}$ where $l$ is the dimension of the output set
\end{theorem*}
The second result applies to more general games and doesn't depend on the input set dimension. On the other hand, the first result has a better dependance in $\eps$.

\subsection{Contribution}
In this paper, we simplify, improve and extend our previous work \cite{CS14}, inspiring ourselves from the techniques used in \cite{JPY13} and blending them with our own. Our main contributions are the following: (1) we present a new parallel repetition theorem for \emph{free games} that improves on the results of both \cite{JPY13} and \cite{CS14} (2) we present a stronger parallel repetition theorem for \emph{free projection games}. 

\paragraph{Parallel repetition theorem for entangled \emph{free games}}
We first show the following:
\begin{theorem}\label{Theorem:ParallelRepetition}
For any free game $G$, we have $\omega^*(G) \le (1 - \eps^2)^{\Omega(\frac{n}{\log(l)})}$.
\end{theorem}

The proof will have two main components. First, as in \cite{CS14}, we use the notion of the superposed information cost to lower bound the value of an entangled game. Informally, the superposed information cost (SIC) of a game represents the minimal amount of information that Alice and Bob must have about each other's classical inputs in other to win the game with probability $1$, while having their inputs in a quantum superposition. In \cite{CS14}, we showed that $SIC(G) \ge \Omega(\eps)$. In this paper, we reprove this statement by simplifying the previous proof. 

We proceed to show that $SIC(G^n) \ge \Omega(n\eps)$. Then, we show that Alice and Bob can win a weaker version of $G^n$, where we only require Alice and Bob to win most games, while having only $\approx O(-\log(\omega^*(G^n)) \frac{\log(l)}{\eps})$ information about each other's inputs in this superposed setting. This is will be done via a communication protocol that will help Alice and Bob win $G^n$. We finally manage to combine these two results to show that 
$-\log(\omega^*(G^n)) \ge \Omega(\frac{n \eps^2}{\log(l)})$ or equivalently $\omega^*(G^n) \le (1 - \eps^2)^{\Omega(\frac{t}{\log(l)})}$.

\paragraph{Parallel repetition theorem for entangled \emph{free projection games}}
We then improve the above theorem for the special case of entangled free projection games.
\begin{theorem}\label{Theorem:ParallelRepetitionProjectionGames}
For any free projection game $G$, we have $\omega^*(G) \le (1 - \eps)^{\Omega(n)}$.
\end{theorem}

The theorem follows by an improvement of the communication protocol mentioned above, for the specific case of free projection games.

\subsection{Organization of the paper}
The rest of the paper is organized as follows. 
In Section \ref{Section:Prel} we introduce some preliminaries concerning quantum information theory. We also present entangled games and define the notion ot the superposed information cost. In Section \ref{Section:Relating} we prove the relation between the superposed information cost and the value of the game. Then, in Section \ref{Section:Proof} we provide the proof of our main result. The organization of the proof is detailed at the beginning of the section. Finally, in Section \ref{Section:Projection} we extend our result to projection games.

\section{Preliminaries}
\label{Section:Prel}

\subsection{The fidelity of two quantum states.}\label{HowClose}
We start by stating a few properties of the fidelity $F$ between two quantum states. 
\begin{definition}
For any two states $\rho,\sigma$, their fidelity $F$ is given by
$F(\rho,\sigma) = F(\sigma,\rho)= \Tr(\sqrt{\rho^{\frac{1}{2}}\sigma\rho^{\frac{1}{2}}}) =  ||\sqrt{\rho}\sqrt{\sigma}||_1$. We also define $\Fbar(\rho,\sigma) = 1 - F(\rho,\sigma)$.
\end{definition}

\begin{fact}\label{POVMfidelity}
For any two states $\rho,\sigma$, and a POVM $E = \{E_1,\dots,E_m\}$ with $p_i = \Tr(\rho E_i)$ and $q_i = \Tr(\sigma E_i)$, we have $F(\rho,\sigma) \le \sum_{i} \sqrt{p_i q_i}$. There exists a POVM for which this inequality is an equality.
\end{fact}

\begin{definition}
A pure state $\ket{\psi}$ in $\spa{A} \otimes \spa{B}$ is a purification of some state $\rho$ in $\spa{B}$ if $\ \Tr_{\spa{A}}(\ketbra{\psi}{\psi}) = \rho$.
\end{definition}

\begin{fact}[Uhlmann's theorem]
For any two quantum states $\rho,\sigma$ and any purification $\ket{\phi}$ of $\rho$, there exists a purification $\ket{\psi}$ of $\sigma$ such that $|\braket{\phi}{\psi}| = F(\rho,\sigma)$.
\end{fact}

\begin{fact}\label{CPTPfidelity}
For any two quantum states $\rho,\sigma$ and a completely positive trace preserving operation $Q$, we have $F(\rho,\sigma) \le F(Q(\rho),Q(\sigma))$.
\end{fact}

\begin{fact}[\cite{SR01,NS03}]~\label{Prop:FidelityInequality}
For any two quantum states $\rho, \sigma$ 
\[
    \max_{\xi} \left( F^2(\rho, \xi) + F^2(\xi, \sigma) \right)
    = 1 + \F(\rho,\sigma).
\]
\end{fact}

As a corollary of Fact~\ref{Prop:FidelityInequality}, we can show a \emph{weak triangle inequality} for the quantity $1-F$.
\begin{proposition}\label{Prop:FidelityInequality2}
For any 3 quantum states $\rho_1,\rho_2,\rho_3$, we have 
\begin{align*}
1 - F(\rho_1,\rho_3) \le  2(1 - F(\rho_1,\rho_2)) + 2(1 - F(\rho_2,\rho_3)).
\end{align*}
\end{proposition}
\begin{proof}
Using Fact~\ref{Prop:FidelityInequality}, we have
\begin{align*}
1 + F(\rho_1,\rho_3) & = \max_{\xi} \left( F^2(\rho_1, \xi) + F^2(\xi, \rho_3) \right) \\
& \ge  F^2(\rho_1, \rho_2) + F^2(\rho_2, \rho_3),
\end{align*}
which gives
\begin{align*}
1 - F(\rho_1,\rho_3) \le 1 - F^2(\rho_1,\rho_2) + 1 - F^2(\rho_2,\rho_3) \le 2(1 - F(\rho_1,\rho_2)) + 2(1 - F(\rho_2,\rho_3)).
\end{align*}
\end{proof}

\begin{definition}
For any two states $\rho,\sigma$, we define $Angle(\rho,\sigma) = Arccos(F(\rho,\sigma))$. Angle is a distance for quantum states \cite[page 413]{NC00}.
\end{definition}

\begin{claim} \label{SuperClaim}
For any 4 quantum states $\rho_1,\rho_2,\rho_3,\rho_4$, we have $$
\Fbar(\rho_1,\rho_4) \le 3(\Fbar(\rho_1,\rho_2) + \Fbar(\rho_2,\rho_3) + \Fbar(\rho_3,\rho_4)).
$$
\end{claim}
\begin{proof}
Let $\alpha = Angle(\rho_1,\rho_4)$. Let also $\alpha_1 = Angle(\rho_1,\rho_2)$, $\alpha_2 = Angle(\rho_2,\rho_3)$, $\alpha_3 = Angle(\rho_3,\rho_4)$. Since $Angle$ is a distance on quantum states, we have $\alpha \le \alpha_1 + \alpha_2 + \alpha_3$. We have 
$$
1 - \cos(\alpha) \le 9(1 - \cos(\alpha/3)) \le 3(1 - \cos(\alpha_1) + 1 - \cos(\alpha_2) + 1 - \cos(\alpha_3)),
$$
where the first inequality can be shown analytically and the second one comes from convexity of the function $1 - \cos$.  From there, we conclude 
$$
\Fbar(\rho_1,\rho_4) \le 3(\Fbar(\rho_1,\rho_2) + \Fbar(\rho_2,\rho_3) + \Fbar(\rho_3,\rho_4)).
$$
\end{proof}

\begin{proposition}\label{Prop:FidelityLast}
For two quantum states $\rho = \sum_{x} p_x \altketbra{x} \otimes \rho_x$ and $\rho' = \sum_x p'_x \altketbra{x} \otimes \rho'_x$, we have $F(\rho,\rho') = \sum_x \sqrt{p_x p_{x'}} F(\rho_x,\rho_{x'})$.
\end{proposition}
\begin{proof}
We use the following definition  of the fidelity: $F(\rho,\rho') = ||\sqrt{\rho}\sqrt{\rho'}||_1$. From there, we immediately have that 
\[ F(\rho,\rho') = \sum_{x} \sqrt{p_x p_{x'}} ||\sqrt{\rho_x}\sqrt{\rho'_x}||_1 = \sum_{x} \sqrt{p_x p_{x'}} F(\rho_x,\rho_{x'}). \]
\end{proof}

\subsection{Information Theory}\label{Section:InformationTheory}
\paragraph{Quantum registers and measured quantum registers}
For a quantum state $\rho$ and a quantum register $X$, we will write $\rho^X$ the reduced state of $\rho$ on register $X$.
For a quantum register $X$, $\X$ corresponds to this register after it was measured in the computational basis. For example, for a quantum pure state $\ket{\phi} = \sum_{x} \sqrt{p_x} \ket{x}_X \otimes \ket{Z_x}_Z$, we have $\ket{\phi}^X = Tr_Z \ket{\phi}$ and $\ket{\phi}^{\X} = \sum_x p_x \kb{x}$. \\

For a quantum state $\rho$, the entropy of $\rho$ is $
S(\rho) = - \Tr(\rho \log(\rho))$.
For a quantum state $\rho \in \spa{X} \otimes \spa{Y}$, $S(X)_\rho$ is the entropy of the quantum register in the space $\spa{X}$ when the total underlying state is $\rho$. In other words, $S(X)_\rho = S(\rho^X) = S(\Tr_{\spa{Y}} (\rho)) $.

$S(X|Y)_{\rho} = S(XY)_{\rho} - S(Y)_{\rho}$ is the conditional entropy of $X$ given $Y$ on $\rho$ and $I(X:Y)_\rho = S(X)_\rho + S(Y)_\rho - S(XY)_\rho $ is the mutual information between $X$ and $Y$ on $\rho$.

For a pair of quantum states $\rho,\sigma$, the relative entropy of $\rho$ with respect to $\sigma$ is defined by $S(\rho || \sigma) = \Tr(\rho\log(\rho)) - \Tr(\rho\log(\sigma))$. It can be shown that $I(X:Y)_{\rho} = S(\rho^{XY} || \rho^X \otimes \rho^Y)$.

The min-relative entropy of $\rho$ with respect to $\sigma$ is defined by $S_{\infty}(\rho || \sigma) = \min\{k: \rho \le 2^k \sigma\}$. 

\begin{fact}[Subadditivity of the conditional entropy]
\begin{align*}
S(AB|C) \le S(A|C) + S(B|C)
\end{align*}
\end{fact}

\begin{fact}[\cite{JPY13}]\label{Claim:DC}
$
S(\rho || \sigma) \ge 1 - F(\rho,\sigma)$.
This immediately implies $I(X:Y)_{\rho} \ge 1 - F(\rho,\rho^X \otimes \rho^Y)$.
\end{fact}

\begin{proposition}
Let $\sigma^{12}, \rho^1, \rho^2$ three classical states. We have 
$$
S(\sigma^{12} || \rho^1 \otimes \rho^2) \ge S(\sigma^1 || \rho^1) + S(\sigma^2 || \rho^2)
$$
\end{proposition}
\begin{proof}
We write $\sigma^{12} = \sum_x q_x \altketbra{x} \otimes \sigma^2_x$. Using the chain rule for relative entropy, we have 
\begin{align*}
S(\sigma^{12} || \rho^1 \otimes \rho^2) & = S(\sigma^1 || \rho^1) + \E_{x \leftarrow q_x} S(\sigma^2_x || \rho^2) \\
& \ge S(\sigma^1 || \rho^1) + S(\E_{x \leftarrow q_x}  \sigma^2_x || \rho^2) \\
& = S(\sigma^1 || \rho^1) + S(\sigma^2 || \rho^2).
\end{align*}
\end{proof}
\begin{corollary}\label{Corollary:RelativeEntropySuperadditivity}
Let $\sigma^{Z}$ and $\rho^Z$ some classical distribution with $Z = Z_1 \otimes \dots \otimes Z_n$ and $\rho^Z = \rho^{Z_1} \otimes \dots \otimes \rho^{Z_n}$. We have
$S(\sigma^Z || \rho^Z) \ge \sum_i S(\sigma^{Z_i} || \rho^{Z_i})$.
\end{corollary}

The following facts were used in \cite{JPY13}.
\begin{fact} \label{f7}
$
S_{\infty} (\rho || \sigma) \ge S(\rho || \sigma).$
\end{fact}
\begin{fact} \label {f8}$
S(\rho^{XY} || \rho^X \otimes \rho^Y) \le S(\rho^{XY} || \sigma^X \otimes \sigma^Y)$
for any $\rho,\sigma$. 
\end{fact}
\begin{fact}\label{Claim:RelEntTrace}
For any states $\rho,\sigma$ each in space $\spa{XY}$, we have $ S(\rho || \sigma) \ge S(\rho^X || \sigma^X)$.
\end{fact}

\begin{proposition}
For any pure state $\ket{\phi}$ in $\spa{A} \otimes \spa{B}$, we have 
$$
\kb{\phi} \le |B|^2 (\kb{\phi^A} \otimes \kb{\phi^B}).$$
\end{proposition}
\begin{proof}
We write $\ket{\phi} = \sum_{i = 1}^{|B|} \sqrt{p_i} \ket{e_i}\ket{f_i}$ a Schmidt decomposition of $\ket{\phi}$. We have $\kb{\phi^A} = \sum_i p_i \kb{e_i}$ and $\kb{\phi^B} = \sum_i p_i \kb{f_i}$. We have
\begin{align*}
\bra {\phi} \cdot (\kb{\phi^A} \otimes \kb{\phi^B}) \cdot \ket{\phi} = 
\sum_{i,j =1}^{|B|} p_i p_j \bra{\phi} \cdot (\kb{e_i} \otimes \kb{f_j}) \cdot \ket{\phi} = \sum_{i = 1}^{|B|} p_i^3 \ge \frac{1}{|B|^2},
\end{align*}
which implies $\kb{\phi^A} \otimes \kb{\phi^B} \ge \frac{1}{|B|^2} \kb{\phi}$.
\end{proof} 
\begin{corollary} \label{Corollary:CoolProduct}
For any state $\rho$ in $\spa{A} \otimes \spa{B}$ with $|A| \ge |B|$, we have 
$$
\rho \le |B|^2 (\rho^A \otimes \rho^B).$$
\end{corollary}
\begin{proof}
Fix a state $\rho$ in $\spa{A} \otimes \spa{B}$ and a purification $\ket{\phi}$ in $\spa{Z} \otimes \spa{A} \otimes \spa{B}$ of $\rho$. From the previous proposition, we have 
$$
\kb{\phi} \le |B|^2 (\kb{\phi^{ZA}} \otimes \kb{\phi^B}). $$
We trace out the $Z$ part to each side and we obtain
$$
\rho \le |B|^2 (\rho^A \otimes \rho^B).
$$
\end{proof}

\COMMENT{
\begin{fact}\label{Fact:InfoInfo}
Let a Hilbert space $U = U_1 \otimes \dots \otimes U_n$. We have 
$I(U:V)_\rho \ge \sum_i I(U_i : V) + S(U)_\rho - \sum_i S(U_i)_{\rho}$.
\end{fact}
\begin{proof}
\begin{align*}
I(U:V)_{\rho} = S(U)_{\rho} - S(U|V)_{\rho} \ge S(U)_{\rho} - \sum_i S(U_i | V)_{\rho} = \sum_i I(U_i : V)_\rho + S(\rho) - \sum_i S(U_i)_{\rho}.
\end{align*}
\end{proof}
}

\subsection{Entangled Games}\label{Section:EntangledGames}
We now define the notion of an entangled game and its value. 
\begin{definition}
An entangled game $G = (I,O,V,p)$ is defined by finite input and output sets $I$ and $O$ as well as an accepting function $V: O^2 \times I^2 \rightarrow \zo$ and a probability distribution $p: I^2 \rightarrow [0,1]$.
\end{definition}

A strategy for the game proceeds as follows. Alice and Bob can share any quantum state. Then, Alice receives an input $x \in I$ and Bob receives an input $y \in I$ where these inputs are sampled according to $p$. They can perform any quantum operation but are not allowed to communicate. Alice outputs $a \in O$ and Bob outputs $b \in O$. They win the game if $V(a,b|x,y) = 1$. 

The \emph{entangled value} of a game $G$ is the maximal probability with which Alice and Bob can win the game. From standard purification techniques, we can assume that w.l.o.g., Alice and Bob can share a pure state $\ket{\phi}$. Moreover, their optimal strategy can be described as projective measurements $A^x = \{A^x_a\}_{a \in O}$ and $B^y = \{B^y_b\}_{b \in O}$ on $\ket{\phi}$.

This means that after receiving their inputs, they share a state of the form
\[
\rho = \sum_{x,y \in I} p_{xy} \altketbra{x} \otimes \altketbra{\phi} \otimes \altketbra{y}, 
\]
for some state $\ket{\phi}$. 
\begin{definition}
The entangled value of a game $G$ is 
\[
\omega^*(G) = \sup_{\ket{\phi},A^x,B^y} \sum_{x,y,a,b} p_{xy} V(a,b|x,y) \triple{\phi}{A^x_a \otimes B^y_b}{\phi}.
\]
\end{definition}

\begin{definition}
A game $G = (I,O,V,p)$ is called \emph{free} if $p$ is a product distribution.
 \end{definition}

\begin{definition}
A game $G = (I,O,V,p)$ is a \emph{projection game} if $\forall  x,y \in I$ and $\forall  b \in O$, $\exists! \ a$ st. $V(ab | xy) = 1$.
\end{definition}
\subsubsection{Value of a game with advice states}
\label{Section:advice}
Consider a game $G = (I,O,V,p)$. We are interested in the value of the game when the two players share an advice state $\ket{\phi_{xy}}$ on inputs $x,y$. This means that Alice and Bob share a state of the form 
\[
\rho = \sum_{x,y} p_{xy} \altketbra{x} \otimes \altketbra{\phi_{xy}} \otimes \altketbra{y}.
\]

\begin{definition}
The entangled value of  $G$, given that Alice and Bob share the above state $\rho$ is
\[
\omega^*(G|\rho) = \max_{A^x,B^y} \sum_{x,y,a,b} p_{xy} V(a,b|x,y) \triple{\phi_{xy}}{A^x_a \otimes B^y_b}{\phi_{xy}}.
\]
\end{definition}
\subsubsection{Repetition of entangled games}

In the $n$-fold parallel repetition of a game $G$, each player gets $n$ inputs from $I$ and must produce $n$ outputs from $O$. Each instance of the game will be evaluated as usual by the function~$V$. The players win the parallel repetition game if they win \emph{all} the instances.
More formally, for a game $G = (I,O,V,p)$ we define $G^n = (I',O',V',q)$, where
$I' = I^{\times n}, O' = O^{\times n}, q_{xy} = \Pi_{i \in [n]} p_{x_i,y_i}$ and $V'(a,b|x,y) = \Pi_{i \in [n]} V(a_i,b_i|x_i,y_i)$. While playing $G^n$, we say that Alice and Bob win game $i$ if $V(a_i,b_i|x_i,y_i) = 1$.

\subsubsection{Majority game}\label{Section:DefinitionMajorityGame}
For a game $G = (I,O,V,p)$ and a real number $\alpha \in [0,1]$ we define $G^n_\alpha = (I',O',V',p')$ as follows: $I' = I^{\times n}$, $O' = O^{\times n}$, $p'_{xy} = \Pi_{i \in [n]} p_{x_i,y_i}$ as in $G^n$. We define $V'$ as follows:
\begin{align*}
V'(a,b|x,y) = 1 \Leftrightarrow \#\{i: V(a_i,b_i|x_i,y_i) = 1\} \ge \alpha n.
\end{align*}

\subsection{Definition of the superposed information cost}
Informally, the superposed information cost (SIC) of a game represents the minimal amount of information that Alice and Bob must have about each other's classical input register in other to win the game with probability $1$, while having their own inputs in a quantum superposition. More formally:

\begin{definition}
Fix a game $G = (I,O,V,p)$.
\[ 
SIC(G) = \min_{\ket{\Omega}} \ I(\X : BY)_{\ket{\Omega}} + I(\Y : XA)_{\ket{\Omega}},
\]
where the minimum is taken over all $\ket{\Omega} = \sum_{x,y} \sqrt{p_{xy}} \ket{x}_{X} \ket{\phi_{xy}}_{AB} \ket{y}_{Y}$ such that $\omega^*(G | \rho) = 1$ with $\rho = \sum_{xy} p_{xy} \altketbra{x} \otimes \altketbra{\phi_{xy}} \otimes \altketbra{y}$. Recall that $\X$ (resp. $\Y$) corresponds to the $X$ (resp. $Y$) register measured in the computational basis. $\X$ and $\Y$ correspond to Alice's and Bob's classical inputs.
\end{definition}

We also generalize the above definition to the case where we minimize over all states such that $\omega^*(G) = \alpha$.
\begin{definition}
Fix a game $G = (I,O,V,p)$.
\[ 
SIC(G,\alpha) = \min_{\ket{\Omega}} \ I(\X : BY)_{\ket{\Omega}} + I(\Y : XA)_{\ket{\Omega}},
\]
where the minimum is taken over all $\ket{\Omega} = \sum_{x,y} \sqrt{p_{xy}} \ket{x}_{X} \ket{\phi_{xy}}_{AB} \ket{y}_{Y}$ such that $\omega^*(G | \rho) = \alpha$ with $\rho = \sum_{xy} p_{xy} \altketbra{x} \otimes \altketbra{\phi_{xy}} \otimes \altketbra{y}$.
\end{definition}
Notice that we have by definition $SIC(G,1) = SIC(G)$ and $SIC(G,\omega^*(G)) = 0$.

\section{Relating $SIC(G)$ and $\omega^*(G)$} \label{Section:Relating}
Our goal here is to lower bound the superposed information cost of $G$ in termes of its entangled value. In this Section, we show that for any game $G$, $SIC(G) \ge \Omega(\eps)$ where $\eps = 1 - \omega^*(G)$. We are actually able to make that result robust in the following way: for any fixed constant $\gamma < 1$, we can show that $SIC(G,1 - \gamma \eps) \ge \Omega(\eps)$. Moreover, we will also extend this to case of a game $H$ which is close to a free game.

In order to prove this, we show in Section \ref{Section:Dependency} that for any state $\ket{\Omega} = \sum_{xy} \ket{x}_X \ket{\phi_{xy}}_{AB} \ket{y}_Y$, if the quantity  $I(\X:AB)_{\ket{\Omega}} + I(\Y:XA)_{\ket{\Omega}}$ is small then Alice and Bob can almost remove the dependency in $x,y$ of the advice states $\ket{\phi_{xy}}$ by local quantum isometries, using only their input registers as control bits. This statement actually requires Alice and Bob to have a quantum superposition of their inputs and would not be true if they both had classical inputs instead. Then, in Section \ref{Section:ProofOfRelation}, we show how to use the above quantum isometries to bound the superposed information cost.

\subsection{Removing the dependence on the inputs from the advice states}\label{Section:Dependency}

Consider a game with advice, with initial state $\ket{\Omega_0} = \sum_{xy} \sqrt{p_{xy}} \ket{x}_{\spa{X}} \otimes \ket{\phi_{xy}}_{\spa{AB}} \otimes \ket{y}_{\spa{Y}}$. We first show that if the advice states $\{\ket{\phi_{xy}}_{\spa{AB}}\}$ do not give Alice and Bob much information about each other's input registers then Alice can perform a local operation to almost decouple the advice states with his input register. By symmetry, Bob can do the same.
We combine these two facts in Proposition \ref{Proposition:LocalUnitaries}: Alice and Bob can perform local operations such that the resulting advice states are close to $\ket{\psi}$, which is independent of $x,y$.

\begin{lemma}
Let $\ket{\Omega_0} = \sum_{xy} \sqrt{p_{xy}} \ket{x}_{\spa{X}} \otimes \ket{\phi_{xy}}_{\spa{AB}} \otimes \ket{y}_{\spa{Y}}$. If $I(\X:BY)_{\ket{\Omega_0}} \le \delta$ then there exist quantum isometries $U_x$ from $\spa{A}$ to $\spa{A'}$ such that $\Fbar(\ket{\Omega_1},\ket{\Omega_1}^{X} \otimes \ket{\Omega_1}^{A'BY}) \le 9\delta$ with $\ket{\Omega_1} = \sum_{xy} \sqrt{p_{xy}} \ket{x} \otimes (U_x \otimes I_B) \ket{\phi_{xy}} \otimes \ket{y}.$
\end{lemma}
\begin{proof}
Let $\rho_x$ be the state in $\spa{BY}$ when Alice measures the $\spa{X}$ register in the computational basis and observes $x$. Let also $\rho_+ = \sum_x p_{x\cdot}  \rho_x = \kb{\Omega_0}^{BY}$. We have 
\begin{align*}
\delta \ge I(\X:BY)_{\ket{\Omega_0}} & \ge 1 - F(\kb{\Omega_0}^{\X BY}, \kb{\Omega_0}^{\X} \otimes 
\kb{\Omega_0}^{BY}) \\
& =  1 - F(\sum_x p_{x\cdot} \kb{x} \otimes \rho_x, \sum_x p_{x\cdot} \kb{x} \otimes \rho_+)
= 1 - \sum_{x} p_{x\cdot} F(\rho_x,\rho_+),
\end{align*}
where the first inequality comes from Fact \ref{Claim:DC}.

Let $\ket{\Phi_y} = \sum_y \sqrt{p_{\cdot y}} \ket{\phi_y}_{\spa{A'B}} \ket{y}_{\spa{Y}}$ be a purification of $\rho_+$ in $\spa{A'BY}$ for some $\ket{\phi_y}$ with $|A'| \ge |A|$. Let also $\ket{\Psi_{xy}} = \sum_{y} \sqrt{p_{\cdot y}} \ket{\psi_{xy}}_{\spa{AB}} \otimes \ket{y}$ which is a purification of $\rho_x$.  By Uhlmann's theorem, we consider quantum isometries $U_x$ from $\spa{A}$ to $\spa{A'}$ such that $\bra{\Phi_y} (U_x \otimes I_{BY}) \ket{\Psi_{xy}} = F(\rho_x,\rho_+)$. We  also define 
\begin{itemize}
\item $\ket{\Omega_1} = \sum_{x} \sqrt{p_{x\cdot}} \ket{x}_{\spa{X}} \otimes (U_x \otimes I_{\spa{BY}}) \ket{\Psi_{xy}}$
\item $\ket{\Omega'_1} = \sum_{x} \sqrt{p_{x\cdot}} \ket{x}_{\spa{X}} \otimes \ket{\Phi_{y}}$
\end{itemize}
We have
\begin{align*}
\braket{\Omega_1}{\Omega'_1} = \sum_x p_{x\cdot} \bra{\Phi_y} (U_x \otimes I_{BY}) \ket{\Psi_{xy}}
= \sum_x p_{x\cdot} F(\rho_x,\rho_+) \ge 1 - \delta.
\end{align*}
or equivalently $\Fbar(\ket{\Omega_1},\ket{\Omega'_1}) \le \delta$.
Notice also that $\ket{\Omega'_1} = \ket{\Omega'_1}^X \otimes \ket{\Omega'_1}^{A'BY}$
From there, we have 
\begin{itemize}
\item $\Fbar(\ket{\Omega_1},\ket{\Omega'_1}^X \otimes \ket{\Omega'_1}^{ABY}) \le \delta$,
\item $\Fbar(\ket{\Omega'_1}^X \otimes \ket{\Omega'_1}^{ABY}, \ket{\Omega_1}^X \otimes \ket{\Omega'_1}^{ABY}) = \Fbar(\ket{\Omega'_1}^X, \ket{\Omega_1}^X) \le \delta$,
\item $\Fbar(\ket{\Omega_1}^X \otimes \ket{\Omega'_1}^{ABY}, \ket{\Omega_1}^X \otimes \ket{\Omega_1}^{ABY}) \le \Fbar(\ket{\Omega'_1}^{ABY}, \ket{\Omega_1}^{ABY}) \le \delta$.
\end{itemize}
We now use Claim \ref{SuperClaim} from Section \ref{HowClose}, which states 
that for any 4 quantum states $\rho_1,\rho_2,\rho_3,\rho_4$, we have $
\Fbar(\rho_1,\rho_4) \le 3(\Fbar(\rho_1,\rho_2) + \Fbar(\rho_2,\rho_3) + \Fbar(\rho_3,\rho_4)).
$
We take $\rho_1 = \kb{\Omega_1}$, $\rho_2 =  \kb{\Omega'_1}$, $\rho_3 = \ket{\Omega_1}^X \otimes \ket{\Omega'_1}^{ABY}$ and $\rho_4 = \ket{\Omega_1}^X \otimes \ket{\Omega_1}^{ABY}$.
We conclude that 
$$\Fbar(\rho_1,\rho_4) = \Fbar(\ket{\Omega_1},\ket{\Omega_1}^X \otimes \ket{\Omega_1}^{ABY}) \le 3(3\delta) = 9\delta.$$
\end{proof}
Similarly, we can prove the following.
\begin{lemma}
Let $\ket{\Omega_0} = \sum_{xy} \sqrt{p_{xy}} \ket{x}_{\spa{X}} \otimes \ket{\phi_{xy}}_{\spa{AB}} \otimes \ket{y}_{\spa{Y}}$. 
If $I(\Y:XA)_{\ket{\Omega_0}} \le \delta$ then there exist quantum isometries $V_y$ from $\spa{B}$ to $\spa{B'}$ such that $\Fbar(\ket{\Omega_2},\ket{\Omega_2}^{XAB'} \otimes \ket{\Omega_2}^{Y}) \le 9\delta$ with $\ket{\Omega_2} = \sum_{xy} \sqrt{p_{xy}} \ket{x} \otimes (I_A \otimes V_y) \ket{\phi_{xy}} \otimes \ket{y}.$
\end{lemma}

We now combine the two lemmata above.

\begin{proposition}\label{Proposition:LocalUnitaries}
Let $\ket{\Omega_0} = \sum_{xy} \sqrt{p_{xy}} \ket{x}_{\spa{X}} \otimes \ket{\phi_{xy}}_{\spa{AB}} \otimes \ket{y}_{\spa{Y}}$. If $I(\X:BY)_{\ket{\Omega_0}} \le \delta$ and $I(\Y:XA)_{\ket{\Omega_0}} \le \delta$ then there exist quantum isometries $U_x$ and $V_y$, respectively from $A$ to $A'$ and from $B$ to $B'$, such that $\Fbar(\ket{\Omega_3},\ket{\Omega_3}^{XY} \otimes \ket{\Omega_3}^{A'B'}) \le 81\delta$ with $\ket{\Omega_3} = \sum_{xy} \sqrt{p_{xy}} \ket{x} \otimes (U_x \otimes V_y) \ket{\phi_{xy}} \otimes \ket{y}$.
\end{proposition}
\begin{proof}
We consider the quantum isometries $U_x,V_y$ from the previous two lemmata as well as the states $\ket{\Omega_1},\ket{\Omega_2}$. Since you can from $\ket{\Omega_1}$ (resp. $\ket{\Omega_2}$) to $\ket{\Omega_3}$ by a quantum isometry not acting on $X$ (resp. $Y$), we have 

$$\Fbar(\ket{\Omega_3},\ket{\Omega_3}^X \otimes \ket{\Omega_3}^{A'B'Y}) = 
\Fbar(\ket{\Omega_1},\ket{\Omega_1}^X \otimes \ket{\Omega_1}^{A'BY}) \leq 9\delta
$$

and 

$$\Fbar(\ket{\Omega_3},\ket{\Omega_3}^{XA'B'} \otimes \ket{\Omega_3}^{Y}) = 
\Fbar(\ket{\Omega_2},\ket{\Omega_2}^{XAB'} \otimes \ket{\Omega_2}^{Y}) \leq 9\delta.
$$

From there, we obtain:
\begin{itemize}
\item $\Fbar(\ket{\Omega_3},\ket{\Omega_3}^X \otimes \ket{\Omega_3}^{A'B'Y}) \le 9\delta$,
\item  $\Fbar(\ket{\Omega_3}^X \otimes \ket{\Omega_3}^{A'B'Y}, \ket{\Omega_3}^X \otimes \ket{\Omega_3}^{A'B'} \otimes \ket{\Omega_3}^{Y}) = \Fbar(\ket{\Omega_3}^{A'B'Y}, \ket{\Omega_3}^{A'B'} \otimes \ket{\Omega_3}^{Y})  \\  \le 
\Fbar(\ket{\Omega_3}, \ket{\Omega_3}^{XA'B'} \otimes \ket{\Omega_3}^{Y}) \le 9\delta,$
\item $\Fbar(\ket{\Omega_3}^X \otimes \ket{\Omega_3}^{A'B'} \otimes \ket{\Omega_3}^{Y},\ket{\Omega_3}^{XY} \otimes \ket{\Omega_3}^{A'B'}) = \Fbar(\ket{\Omega_3}^{X} \otimes \ket{\Omega_3}^{Y},\ket{\Omega_3}^{XY}) \\ \le 
\Fbar(\ket{\Omega_3}^{XA'B'} \otimes \ket{\Omega_3}^{Y} , \ket{\Omega_3}) \le 9 \delta.$
\end{itemize}
Using again Claim \ref{SuperClaim}, we conclude that $\Fbar(\ket{\Omega_3},\ket{\Omega_3}^{XY} \otimes \ket{\Omega_3}^{A'B'}) \le 3(3 \cdot 9\delta) = 81\delta$.
\end{proof}

\subsection{Proving the relation}\label{Section:ProofOfRelation}

We are now ready to relate the superposed information cost and the value of an entangled game. To do this, we consider the above results on removing the dependence on the inputs, and this time we work on advice states that allow players to win the game.

\begin{proposition} \label{Proposition:SICnGame}
For any game $G$ with $\omega^*(G) = 1 - \eps$, we have
$$SIC(G,1-\delta) \ge \frac{1}{81}\left(1 - \sqrt{(1-\eps)(1-\delta)} - \sqrt{\delta \eps}\right).$$
As special cases, we have $SIC(G) \ge \frac{\eps}{162}$ and $SIC(G,1-\frac{\eps}{8}) \ge \frac{\eps}{324}$.
\end{proposition}
\begin{proof}
Let $\ket{\Omega} = \sum_{xy} \sqrt{p_{xy}} \ket{x} \otimes \ket{\phi_{xy}} \otimes \ket{y}$ such that Alice and Bob can win $G$ with probability $1-\delta$ when sharing states $\ket{\phi_{xy}}$ and $I(\X:BY)_{\ket{\Omega}} + I(\Y:XA)_{\ket{\Omega}} = SIC(G,1-\delta)$. From Proposition \ref{Proposition:LocalUnitaries}, we consider quantum isometries $U_x$ and $V_y$ acting respectively from $\spa{A}$ to $\spa{A}'$ and from  $\spa{B}$ to $\spa{B}'$ and the state $\ket{\Omega_3} = \sum_{xy} \sqrt{p_{xy}} \ket{x} \otimes (U_x \otimes V_y) \ket{\phi_{xy}} \otimes \ket{y}$ such that $\Fbar(\ket{\Omega_3},\ket{\Omega_3}^{XY} \otimes \ket{\Omega_3}^{AB}) \le 81 \cdot SIC(G,1-\delta)$.

Notice that Alice and Bob can locally win $G$ with probability $1-\delta$ when sharing $\ket{\Omega_3}$ and measuring the input registers since they can recreate $\ket{\phi_{xy}}$ using local quantum operations. On the other hand, this strategy will only succeed with probability at most $\omega^*(G)$ when sharing $\ket{\Omega_3}^{XY} \otimes \ket{\Omega_3}^{AB}$. 

Let $\rho_{win}$ and $\rho_{lose}$ denote the final states in case of victory of loss, respectively. It follows from the above discussion that
\begin{align*}
\Fbar(\ket{\Omega_3},\ket{\Omega_3}^{XY} \otimes \ket{\Omega_3}^{AB}) & \ge 
\Fbar((1 - \delta) \rho_{win} + \delta \rho_{lose},(1-\eps) \rho_{win} + \eps \rho_{lose}) \\
&  = 1 - \sqrt{(1 - \eps)(1-\delta)} - \sqrt{\delta \eps},
\end{align*}
 which proves the main statement. The two special cases follow from this inequality.
\end{proof}

We now prove that a similar statement still holds if we replace the input distribution $p$ with a slightly perturbed version $q$. The perturbation is quantified in terms of the relative entropy $S( q || p)$.

\begin{lemma} \label{Lemma:RelEnt}
Let $G = (I,O,V,p)$ such that $\omega^*(G) = 1-\eps$. Let $H = (I,O,V,q)$ such that $S(q || p) \le \frac{\eps}{8}$. We have $\omega^*(H) \le 1 - \frac{\eps}{4}$.
\end{lemma}
\begin{proof}
We have that
$
\frac{\eps}{8} \ge S(q || p) \ge \Fbar(q,p).
$
Let $\ket{\phi}$ be the shared state that allows Alice and Bob to win $H$ with probability $\omega^*(H)$.
Let $\rho_p = \sum_{xy} p_{xy} \kb{x} \otimes \kb{\phi} \otimes \kb{y}$ and $\rho_q = \sum_{xy} q_{xy} \kb{x} \otimes \kb{\phi} \otimes \kb{y}$. If Alice and Bob apply the optimal strategy to win $H$ on $\rho_q$, they win with probability $\omega^*(H)$ while they win with probability at most $\omega^*(G)$ on $\rho_p$. 
Let $\rho_{win}$ and $\rho_{lose}$ denote the final states in case of victory of loss, respectively. We have 
\begin{align*}
\frac{\eps}{8} \geq  \Fbar(q,p) = \Fbar (\rho_q,\rho_p)  & \ge \Fbar (\omega^*(H) \rho_{win} + (1 -\omega^*(H)) \rho_{lose}, (1 - \eps) \rho_{win} + \eps \rho_{lose} ) \\
& = 1 - \sqrt{\omega^*(H) (1 - \eps)} - \sqrt{(1 - \omega^*(H))\eps},
\end{align*}
which implies $\omega^*(H) \le 1 - \frac{\eps}{4}$.
\end{proof}
\begin{proposition}\label{MainProposition}
Let $G = (I,O,V,p)$ on a product distribution such that $\omega^*(G) = 1 - \eps$. Let $H = (I,O,V,q)$ such that $S(q || p) \le \frac{\eps}{8}$. We have $SIC(H,1 - \frac{\eps}{32}) \ge \frac{\eps}{1296} = \Omega(\eps)$.
\end{proposition}
\begin{proof}
From Lemma \ref{Lemma:RelEnt}, we know that $\omega^*(H) \le 1 - \frac{\eps}{4}$. By Proposition \ref{Proposition:SICnGame} we have  $SIC(H,1 -\frac{\eps}{32}) \ge \frac{\eps}{1296}$.
\end{proof}

\section{Proving parallel repetition}\label{Section:Proof}

In this section we prove the main result. The proof will proceed as follows.
We fix a free game $G = (I,O,V,p)$ with $\omega^*(G) = 1 - \eps$ and $\omega^*(G^n) = 2^{-t}$ for some $t$. The previous section ended with Proposition \ref{MainProposition} where we showed that $SIC(H,1-\frac{\eps}{32}) \ge \Omega(\eps)$ for any game $H = (I,O,V,q)$ with $S(q || p) \le \frac{\eps}{8}$. 

Here we construct a game $H = (I,O,V,q)$ such that $S(q || p) \le \frac{\eps}{8}$ and $SIC(H,1-\frac{\eps}{32}) \le O(\frac{t \log(l)}{n \eps})$. Combining the inequalities above, we conclude that $t = \Omega(\frac{n \eps^2}{\log(l)})$ or equivalently $\omega^*(G^n) = (1 - \eps^2)^{\Omega(\frac{n}{\log(l)})}$. \\

Our goal is to construct this game $H$ as well as some advice states that will imply $SIC(H,1-\frac{\eps}{32}) \le O(\frac{t \log(l)}{n \eps})$. This Section will be organized as follows
\begin{itemize}
\item In Section \ref{Section:Checking}, we present a classical checking procedure that captures the following idea: if Alice and Bob play $G^n$ according to the optimal strategy  then Bob can know whether they won $G^n$ or not with Alice sending only roughly $O(\frac{t}{\eps})$ bits.
\item In Section \ref{Section:ConstructionAdvice}, we present how to construct these advice states using the checking procedure above.
\item In Section \ref{Section:Index}, we show how to choose a good instance of the game which will characterize $H$ and the advice states.
\item In Section \ref{Section:Conclude}, we show our main Theorem.
\end{itemize} 

\subsection{The checking procedure}\label{Section:Checking}

We consider the following procedure:
\\ \\
\cadre{
\begin{center} Checking procedure \end{center}
\begin{itemize}
\item Alice and Bob share a state $\ket{\phi}$ that allows them to win $G^n$ with probability $\omega^*(G^n) = 2^{-t}$.
\item Alice and Bob get inputs $x = x_1,\dots,x_n$ and $y = y_1,\dots,y_n$, with $x,y \in I^n$ following the distribution of $G^n$, play the game according to the optimal strategy and output $a,b$.
\item Alice and Bob have some shared randomness that correspond to $v$ random indices $i_1,\dots,i_v \in [n]$, where $v$ will be specified later. Let $C$ be this set of indices. For all $i \in C$, Alice sends $x_i,a_i$ to Bob.
\item Bob checks that $\forall i \in C$, $V(a_i b_i | x_i y_i) = 1$. If this holds, we say that Bob succeeds the test. Otherwise, we say that Bob aborts.
\end{itemize}
} $ \ $ \\ \\

We first show the following

\begin{proposition}\label{Proposition:CommuncationProtocol}
If Alice and Bob perform the above protocol with $v = \frac{256}{\eps}\left(t + \log(1/\eps) + 8\right)$, we have:
\begin{enumerate}
\item $Pr[\textrm{Bob succeeds}] \ge 2^{-t}$
\item $Pr[\textrm{Alice and Bob win } \ge (1-\frac{\eps}{256})n \textrm{ games }| \textrm{ Bob succeeds} ] \ge (1-\frac{\eps}{256})$.
\end{enumerate}
where  
$$Pr[\textrm{A\textsc{\&}B win } \ge (1-\frac{\eps}{256})n \textrm{ games } \mid \textrm{Bob succeeds} ] = \Pr[ \#\{i : V(a_ib_i | x_i y_i) = 1\} \ge n(1 - \frac{\eps}{256}) \mid \textrm{Bob succeeds}].$$
\end{proposition}
\begin{proof}
We first have:
\begin{align*}
\Pr[\textrm{Bob succeeds}] & = \Pr[\textrm{Alice and Bob win } G_i \ \forall i \in C] \\ & \ge \Pr[\textrm{Alice and Bob win } G_i \ \forall i \in [n] ] \ge 2^{-t}.
\end{align*}

For a uniformly random index $i$, we have:
\begin{align*}
\Pr[\textrm{Alice and Bob win } G_i \mid \textrm{ Alice and Bob win }  \le (1 - \frac{\eps}{256})n \textrm { games }] \le 1 - \frac{\eps}{256}.
\end{align*}
Since the indices $i_1,\dots,i_v$ are independent random indices in $[n]$, we~have 
\begin{align*}
& \Pr[\textrm{Bob succeeds} | \textrm{ Alice and Bob win } \le (1 - \frac{\eps}{256})n \textrm { games}] \\
& =  \Pr[\textrm{Alice and Bob win } G_i \ \forall i \in C | \textrm{ Alice and Bob win } \le (1 - \frac{\eps}{256})n \textrm { games}] \\ 
& \le (1 - \frac{\eps}{256})^v.
\end{align*}
Next, we have:
\begin{align*}
& \Pr[\textrm{Alice and Bob win } \le (1 - \frac{\eps}{256})n \textrm { games} \mid \textrm{Bob succeeds}] \cdot \Pr[\textrm{Bob succeeds}] \\
& =  \Pr[\textrm{Bob succeeds} \mid \textrm{Alice and Bob win }  \le (1 - \frac{\eps}{256})n \textrm { games}] \cdot \Pr[\textrm{Alice and Bob win }  \le (1 - \frac{\eps}{256})n \textrm { games}] \\
& \le \Pr[\textrm{Bob succeeds} \mid \textrm{Alice and Bob win }  \le (1 - \frac{\eps}{256})n \textrm { games}] \\
& \le (1 - \frac{\eps}{256})^v.
\end{align*}
This gives us:
\begin{align*}
\Pr[\textrm{Alice and Bob win } \le (1 - \frac{\eps}{256})n \textrm { games } | \textrm{ Bob succeeds}] & \le \frac{(1 - \frac{\eps}{256})^v}{\Pr[\textrm{Bob succeeds}]} \\
& \le \frac{(1 - \frac{\eps}{256})^v}{2^{-t}}.
\end{align*}
Since $v = \frac{256}{\eps}(t + \log(1/\eps) + 8)$, we have 
$$
\Pr[\textrm{Alice and Bob win } \le (1 - \frac{\eps}{256})n \textrm { games } | \textrm{ Bob succeeds}] \le \frac{\eps}{256}.
$$
\end{proof}  

\subsection{Constructing the advice states}\label{Section:ConstructionAdvice}
In order to construct the advice states, we perform the above checking procedure but we perform everything in quantum superposition. More precisely, Alice and Bob start with the state 
$$
\ket{\Omega_0} = \sum_{xy} \sqrt{p_{xy}} \ket{x}_{X} \ket{\phi}_{AB} \ket{y}_{Y}.
$$
where $\ket{\phi}$ is the shared state that allows Alice and Bob to win $G^n$ with probability $2^{-t}$. After that, they perform unitarily the strategy to win $G^n$ with this probability $2^{-t}$ without measuring their outputs $a,b$.

Proposition \ref{Proposition:CommuncationProtocol} works for a random $C$. We pick a fixed subset $C$ such that Proposition \ref{Proposition:CommuncationProtocol} holds. Alice sends $x^C$ and $a^C$ to Bob in an extra message register $M_{X^C} \otimes M_{A^C}$.

Let 
$$
\ket{\Omega_1} = \sum_{xy} \sqrt{p_{xy}} \ket{x}_{X} \otimes (\sum_{a,b} \alpha^{xy}_{ab} \ket{a} \ket{\phi^{xy}_{ab}} \ket{b})_{AB} \otimes \ket{y}_{Y} \ket{x^C a^C}_{M_{X^C},M_{A^C}} 
$$
 the state that Alice and Bob share after Alice sends a copy of the registers $X^C,A^C$ to Bob.
Let $\rho = p\cdot \altketbra{\psi} + (1-p) \cdot \altketbra{\psi_{Abort}}$ the state that they share after Bob performs his test. Here, state $\ket{\psi}$ corresponds to the case where Bob succeeds and $\ket{\psi_{Abort}}$ to the case where Bob aborts. We write
$$
\ket{\psi} = \sum_{xy} \sqrt{q_{xy}} \ket{x}_{X} \otimes (\sum_{a,b} \beta^{xy}_{ab} \ket{a} \ket{\psi^{xy}_{ab}} \ket{b})_{AB} \otimes \ket{y}_{Y} \ket{x^C a^C}_{M_{X^C},M_{A^C}}. $$
 From Proposition \ref{Proposition:CommuncationProtocol}, we have $p \ge 2^{-t}$. We define Bob's Hilbert space as $Bob = B \otimes Y \otimes M_{X^C} \otimes M_{A^C}$. Similarly, we will write $Alice = 
 X \otimes A$.  We also write $X = X^C \otimes X^{\overline{C}}$ and $Y = Y^C \otimes Y^{\overline{C}}$.

We now show that $\ket{\psi}$ doesn't give away much information about input registers $X^{\overline{C}}$ and $Y^{\overline{C}}$ to the other player.

\begin{proposition} $$
I(\X^{\overline{C}} : \Bob)_{\ket{\psi}} + I(\Y^{\overline{C}} : \Alice)_{\ket{\psi}} \le 2|M_{A^C}| + 2t \le 2v\log(l) + 2t.$$
\end{proposition}

\begin{proof}
\begin{align*}
\altketbra{\Omega_1}^{\X^{\overline{C}}\Bob} & \le 2^{2|M_{A^C}|} (\altketbra{\Omega_1}^{\X^{\overline{C}}BYM_{X^C}} \otimes \altketbra{\Omega_1}^{M_{A^C}}) \\
& = 2^{2|M_{A^C}||} (\altketbra{\Omega_1}^{\X^{\overline{C}}} \otimes \altketbra{\Omega_1}^{BYM_{X^C}}  \otimes \altketbra{\Omega_1}^{M_{A^C}}).
\end{align*}
The first inequality comes from Corollary \ref{Corollary:CoolProduct} and
the last equality comes from the fact that Bob has no information about $\X^{\overline{C}}$ outside of $M_{A^C}$, since we start from a game on a product distribution. 

Recall that we defined $\rho$ as the state shared by Alice and Bob after Bob performs his test. Since Bob can go from $\ket{\Omega_1}$ to $\rho$ with a local operation on his space, we have:
\begin{align*}
\rho^{\X^{\overline{C}}\Bob} \leq 2^{2|M_{A^C}|} (\rho^{\X^{\overline{C}}} \otimes \rho^{BYM_{\X^C}}  \otimes \rho^{M_{A^C}})\end{align*}
Next, we use $\rho = p\cdot \altketbra{\psi} + (1-p) \cdot \altketbra{\psi_{Abort}}$. We have
$
\altketbra{\psi}^{\X^{\overline{C}}\Bob} \le \frac{1}{p} \rho^{\X^{\overline{C}}\Bob}  \le \frac{1}{p} 2^{2|M_{A^C}|} (\rho^{\X^{\overline{C}}} \otimes \rho^{BYM_{\X^C}}  \otimes \rho^{M_{A^C}}),
$
which gives 
\begin{align*}
S_{\infty}(\altketbra{\psi}^{\X^{\overline{C}}Bob}   \ || \  \rho^{\X^{\overline{C}}} \otimes \rho^{BYM_{\X^C}}  \otimes \rho^{M_{A^C}}) \le 2|M_{A^C}| + \log(1/p).
\end{align*}
Moreover,
\begin{align*}
S_\infty(\altketbra{\psi}^{\X^{\overline{C}}Bob}  \ || \ \rho^{\X^{\overline{C}}} \otimes \rho^{BYM_{\X^C}}  \otimes \rho^{M_{A^C}}) & \ge
S(\altketbra{\psi}^{\X^{\overline{C}}Bob}   \ || \  \rho^{\X^{\overline{C}}} \otimes \rho^{BYM_{\X^C}}  \otimes \rho^{M_{A^C}}) \\ & \ge 
S(\altketbra{\psi}^{\X^{\overline{C}}Bob}   \ || \  \altketbra{\psi}^{\X^{\overline{C}}} \otimes \altketbra{\psi}^{\Bob}) = I(\X^{\overline{C}} : \Bob)_{\ket{\psi}},
\end{align*}
where we use respectively Fact \ref{f7} and Fact \ref{f8}.
Putting this together, we obtain
$$
I(\X^{\overline{C}} : \Bob)_{\ket{\psi}} \le 2|M_{A^C}| + \log(1/p) \le 2|M_{A^C}| + t.
$$
Similarly, we can write 
\begin{align*}
I(\Y^{\overline{C}} : \Alice)_{\ket{\psi}} & = S(\ket{\psi}^{\Y^{\overline{C}}\Alice} \ || \ \ket{\psi}^{\Y^{\overline{C}}} \otimes \ket{\psi}^{\Alice}) 
\le S(\ket{\psi}^{\Y^{\overline{C}}\Alice} \ || \ \rho^{\Y^{\overline{C}}} \otimes \rho^{\Alice}) \\
& \le S_{\infty}(\ket{\psi}^{\Y^{\overline{C}}\Alice} \ || \ \rho^{\Y^{\overline{C}}} \otimes \rho^{\Alice}) 
\le t
\end{align*}
where for the last inequality, we use $\rho^{\Y^{\overline{C}} \Alice} = \rho^{\Y^{\overline{C}}}  \otimes \rho^{\Alice}$ (there is no message from Alice to Bob) and $\kb{\psi} \le 2^{t} \rho$. Putting all this together, we conclude $$
I(\X^{\overline{C}} : \Bob)_{\ket{\psi}} + I(\Y^{\overline{C}} : \Alice)_{\ket{\psi}} \le 2|M_{A^C}| + 2t.
$$
\end{proof}
We now show that on average on $i \in \overline{C}$, Alice and Bob have little information about each other's $i^{th}$ input registers:
\begin{proposition}\label{Proposition:AlmostProduct} $$
\sum_{i \in \overline{C}} I(\X_i : \Bob)_{\ket{\psi}} + I(\Y_i : \Alice)_{\ket{\psi}} \le 2 |M_{A^C}| + 4t$$
\end{proposition}
\begin{proof}
\begin{align*}
\sum_{i \in \overline{C}} I(\X_i : \Bob)_{\ket{\psi}} + I(\Y_i : \Alice)_{\ket{\psi}} & = 
\sum_{i \in \overline{C}} S(\X_i)_{\ket{\psi}} - S(\X_i | Bob)_{\ket{\psi}} + S(\Y_i)_{\ket{\psi}} - S(\Y_i | Bob)_{\ket{\psi}} \\
& \le  \sum_{i \in \overline{C}} S(\X_i)_{\ket{\psi}} - S(\X | Bob)_{\ket{\psi}} + S(\Y_i)_{\ket{\psi}} - S(\Y | Bob)_{\ket{\psi}} \\
& \le 
I(\X^{\overline{C}} : \Bob)_{\ket{\psi}} + I(\Y^{\overline{C}} : \Alice)_{\ket{\psi}} 
+ \sum_{i \in \overline{C}}  S(\X_i)_{\ket{\psi}} - S(\X^{\overline{C}})_{\ket{\psi}} 
\\ & \quad + \sum_{i \in \overline{C}}  S(\Y_i)_{\ket{\psi}} - S(\Y^{\overline{C}})_{\ket{\psi}} \\
& \le 2|M_{A^C}| + 2t
+ \sum_{i \in \overline{C}}  S(\X_i)_{\ket{\psi}} - S(\X^{\overline{C}})_{\ket{\psi}} + 
\sum_{i \in \overline{C}}  S(\Y_i)_{\ket{\psi}} - S(\Y^{\overline{C}})_{\ket{\psi}}.
\end{align*}
Morever, recall that $S_{\infty}(\kb{\psi} \ || \ \rho) \le t$. This gives 
\begin{align*}
t & \ge S_{\infty}(\kb{\psi} \ || \ \rho) \ge S(\kb{\psi} \ || \ \rho)  \ge S(\kb{\psi^{\X^{\overline{C}}}} \ || \ \rho^{\X^{\overline{C}}}) \\
& = S(\kb{\psi^{\X^{\overline{C}}}} \ || \ \bigotimes_{i \in \overline{C}} \rho^{\X_i}).
\end{align*}
where the last equality comes from the face that $\rho^{\X^{\overline{C}}} = \bigotimes_{i \in \overline{C}} \rho^{\X_i}$. Next, we have 
\begin{align*}
S(\kb{\psi}^{\X^{\overline{C}}} \ || \ \bigotimes_{i \in \overline{C}} \rho^{\X^{i}}) & = - S(\X^{\overline{C}})_{\ket{\psi}} - \Tr(\kb{\psi}^{\X^{\overline{C}}} \log(\bigotimes_{i \in \overline{C}} \rho^{\X_{i}})) \\
& = - S(\X^{\overline{C}})_{\ket{\psi}} - \sum_{i \in \overline{C}} \Tr(\kb{\psi}^{\X_i} \log(\rho^{\X_i})) \\
& = - S(\X^{\overline{C}})_{\ket{\psi}} - \sum_{i \in \overline{C}} \Tr(\kb{\psi}^{\X_i} \log(\kb{\psi}^{\X_i})) + \sum_i S(\kb{\psi}^{\X_i} \ || \ \rho^{\X_{i}}) \\
& \ge  - S(\X^{\overline{C}})_{\ket{\psi}} + \sum_{i \in  \overline{C}} S(\X_i)_{\ket{\psi}}.
\end{align*}
From there, we have 
$$ \sum_{i \in  \overline{C}} S(\X_i)_{\ket{\psi}} - S(\X^{\overline{C}})_{\ket{\psi}} \le S(\kb{\psi}^{\X^{\overline{C}}} \ || \ \bigotimes_{i \in \overline{C}} \rho^{\X_{i}})  \le t. $$ 
Similarly, we can show that 
$$ \sum_{i \in  \overline{C}} S(\Y_i)_{\ket{\psi}} - S(\Y^{\overline{C}})_{\ket{\psi}}  \le t. $$
From there, we conclude that 
$$
\sum_{i \in \overline{C}} I(\X_i : \Bob)_{\ket{\psi}} + I(\Y_i : \Alice)_{\ket{\psi}} \le 2 |M_{A^C}| + 4t.$$
\end{proof}

\subsection{Finding a good index}\label{Section:Index}
We consider the states $\ket{\psi}$ and $\rho$ from the previous Section. We now prove that if Alice and Bob share $\ket{\psi}$, there exists an index $i$ such that Alice and Bob can win $G_i$ with high probability but Alice (resp. Bob) doesn't have a lot of information about $y_i$ (resp. $x_i$). We also want that the distribution of inputs $x_i,y_i$ when sharing $\ket{\psi}$ (after conditionning on 'Accept') is close to the distribution of inputs when sharing $\rho$ (before conditionning on 'Accept'). 
\begin{lemma}
We show  the following:
\begin{enumerate}
\item Let $K = \{i : S(\kb{\psi}^{\X_i,\Y_i} || \rho^{\X_i,\Y_i}) \le \frac{4t}{n}\}$, we have $|K| \ge 3n/4$.
\item Let $L = \{i : Pr[\Alice \  \& \ \Bob \mbox{ win } G_i �| \ \mbox{sharing } \ket{\psi}] \ge 1 - \frac{\eps}{32}\}$, we have $|L| \ge 3n/4$.
\item Let $M = \{i \in \overline{C}: S(\X_i : Bob)_{\ket{\psi}} + S(\Y_i : Alice)_{\ket{\psi}} \le \frac{16|M_{A^C}|}{\overline{C}} + \frac{32t}{\overline{C}} \}$. We have $|M| \ge \frac{7\overline{C}}{8}$. In particular,  if $|\overline{C}| \ge 6n/7$, we have $|M| \ge 3n/4$.
\end{enumerate}
If $\overline{C} \ge \frac{6n}{7}$, this implies $|K \cap L \cap M| \ge n/4$. In particular, $K \cap L \cap M \neq \emptyset$.
\end{lemma}
\begin{proof}
For each to these inequalities, we will use the following fact:
\begin{fact} \label{Fact:Sum}
For any $n$ non-negative real numbers $x_i$ with $\frac{1}{n} \sum_{i = 1}^n x_i \le s$, we have $|\{i : x_i \le Cs\}| \ge n(1 - 1/C)$.
\end{fact}
We can now prove our Lemma.
\begin{enumerate}
\item 
Since $\rho = p\cdot \kb{\psi} + (1-p) \cdot \kb{\psi_{Abort}}$, we have 
$S(\kb{\psi} \ || \rho) \le -\log(p) \le t$ which implies from Fact \ref{Claim:RelEntTrace} $S(\kb{\psi}^{\X\Y} || \rho^{\X\Y}) \le t$. Using Corollary \ref{Corollary:RelativeEntropySuperadditivity}, we have $
\sum_{i \in [n]} S(\kb{\psi}^{\X_i\Y_i} || \rho^{\X_i\Y_i}) \le t$ which implies $|K| \ge \frac{3n}{4}$ from Fact \ref{Fact:Sum}.
\item
$\sum_i \Pr[A \ \& \ B \mbox{ win } G_i | \mbox{ sharing } \ket{\psi}]$ is the average number of games that Alice and Bob win when sharing $\ket{\psi}$. From Proposition \ref{Proposition:CommuncationProtocol}, we have $
Pr[\textrm{Alice and Bob win } \ge (1-\frac{\eps}{256})n \textrm{ games }| \textrm{ Bob succeeds} ]
= Pr[\textrm{Alice and Bob win } \ge (1-\frac{\eps}{256})n \textrm{ games }| \textrm{sharing } \ket{\psi} ]
 \ge (1-\frac{\eps}{256})$.

This implies $\frac{1}{n} \sum_i Pr[A \ \& \ B \mbox{ win } G_i | \mbox{ sharing } \ket{\psi}] \ge (1 - \frac{\eps}{256})(1 - \frac{\eps}{256}) \ge 1 - \frac{\eps}{128}$ which gives $|L| \ge 3n/4$.
\item Using Proposition \ref{Proposition:AlmostProduct}, we have
\begin{align*}
\frac{1}{\overline{C}} \sum_{i \in \overline{C}} I(\X_i^{\overline{C}} : \Bob)_{\ket{\psi}} + I(\Y_i^{\overline{C}} : \Alice)_{\ket{\psi}} \le \frac{2|M_{A^C}|}{\overline{C}} + \frac{4t}{\overline{C}}.
\end{align*}
Again, using Fact \ref{Fact:Sum}, we have $|M| \ge \frac{7\overline{C}}{8}$ which implies $|M| \ge 3n/4$ for $|\overline{C}| \ge 6n/7$.
\end{enumerate}
\end{proof}

\subsection{Main result}\label{Section:Conclude}
\begin{theorem}\label{Theorem:ParallelRepetition}
For any free game $G = (I,O,V,p)$, we have $\omega^*(G) \le (1 - \eps^2)^{\Omega(\frac{n}{\log(l)})}$.
\end{theorem}
\begin{proof}
Fix $n$. Let $t$ such that $\omega^*(G^n) = 2^{-t}$. If $t \ge \frac{n\eps}{2048}$ then the statement immediately holds. We now consider the case where $t \le \frac{n\eps}{2048}$. Since 
$|C| = \frac{256}{\eps}(t + \log(1/\eps) + 8)$, we have $|C| \le n/7$ and $|\overline{C}| \ge 6n/7$. 

We consider $\ket{\psi}$ and $\rho$ as defined in Section \ref{Section:ConstructionAdvice}. We pick an element $i \in K \cap L \cap M$. We can find such an $i$ since $K \cap L \cap M \neq \emptyset$. 

We define the game $H = (I,O,V,q)$ where $q = \ket{\psi}^{\X_i,\Y_i}$ is the input distribution of $x_i,y_i$ in state $\ket{\psi}$. Notice that by construction of $\rho$, we have $p = \rho^{\X_i,\Y_i}$ where $p$ is the distribution of game $G$. Since $i \in K$, we have $S(\kb{\psi}^{\X_i,\Y_i} || \rho^{\X_i,\Y_I}) = S(q || p) \le \frac{4t}{n} \le \frac{\eps}{8}$. 

Since $i \in L$, Alice and Bob can win game $i$ (meaning $H$) with probability greater than $1 - \frac{\eps}{32}$ sharing $\ket{\psi}$. We can hence use Proposition \ref{MainProposition} and obtain 
$$
I(\X_i : Bob)_{\ket{\psi}} + I(\Y_i:Alice)_{\ket{\psi}} \ge \Omega(\eps).
$$
Moreover, since $i \in M$, we have 
$$
I(\X_i : Bob)_{\ket{\psi}} + I(\Y_i:Alice)_{\ket{\psi}} \le \frac{32t + 16|M_{A^C}|}{\overline{C}} \le 112 \cdot \frac{2t + v\log(l)}{6n},
$$
with $v = \frac{256}{\eps}\left(t + \log(1/\eps) + 8\right)$. By putting the 2 inequalities together, we have 
$$
\frac{112t}{3n} + \frac{112 \cdot 256 \log(l)}{6 n \eps}\left[t + \log(1/\eps) + 8\right]  \ge \Omega(\eps),$$
which gives $t \ge \Omega(\frac{n \eps^2}{\log(l)})$ and hence, we conclude $\omega^*(G^n) \le (1 - \eps^2)^{\Omega(\frac{n}{\log(l)})}$.
\end{proof}

\section{Extending to free projection games} \label{Section:Projection}

\paragraph{Sketch of proof}
We extend this to the case where in addition, the game we consider is a projection game. This means that for any $x,y,b$, there exists a unique $a$ such that $V(ab|xy) = 1$. The idea of the proof is very similar, the only change is in the communication protocol. Instead of sending $x_i,a_i$ for each $i \in C$, Alice sends all the $x_i$ for $i \in C$ and a hash $h(a^C)$ where $h : [C\log(l)] \rightarrow [2t]$ is taken at random from a universal family of hash functions.

When Bob has $x^C,y^C,b^C$, there exists a unique $a_0^C$ such that $V(a_0^Cb^C | x^Cy^C) = 1$. Bob's check consists of verifying that he receives $h(a_0^C)$. As before, if they win all the games, this test will pass with probability $1$ and $\Pr[\Bob \ \mbox{succeeds}] \ge 2^{-t}$. When calculating $\Pr[\Bob \ \mbox{succeeds} | \mbox{Alice and Bob win} \le n(1-\frac{\eps}{32}) \mbox{games}]$, we have to add the probability that Alice gets $a_1^C \neq a_0^C$ but Bob receives $h(a_1^C) = h(a_0^C)$. Since $h$ is drawn from a universal family of hash functions, this happens with probability at most $2^{-2t}$ which doesn't change fundamentally the analysis. 

The rest is the same except that $|M_{A^C}| = 2t$ instead of $v \log(l)$. By performing the same analysis as before, we obtain 
$$
\omega^*(G^n) \le (1 - \eps)^{\Omega({n})}.
$$  

We now present the full proof, which is very similar to the case of general free games.
\paragraph{Communication protocol} $ \ $ \\ \\ 
\cadre{
\begin{center} Communication protocol \end{center}
\begin{itemize}
\item Alice and Bob share the state $\ket{\phi}$ that allows them to win $G^n$ wp. $\omega^*(G^n)$.
\item Alice and Bob get inputs $x = x_1,\dots,x_n$ and $y = y_1,\dots,y_n$, with $x,y \in I^n$, following the distribution of $G^n$ and play the game according to the optimal strategy and output $a,b$.
\item Alice and Bob have some shared randomness that correspond to $v = O(\frac{t}{\eps})$ random indices $i_1,\dots,i_v \in [n]$. Let $C$ be this set of indices. They also share the description of a hash function $h : [|C| \log(l)] \rightarrow [2t]$ taken randomly from a universal family of hash functions. For all $i \in C$, Alice sends $x_i$ to Bob as well as $h(a^C)$.
\item Since we have a projection game, there exists a unique string $\alpha^C$ such that $\forall i \in C$, $V(\alpha_i b_i | x_i y_i) = 1$. We say that Bob succeeds the test is the string $h(a^C)$ he receives is equal to $h(\alpha^C)$. Otherwise, we say that Bob aborts.
\end{itemize}
} $�\ $ \\

Similarly as in the previous case, we can prove.

\begin{proposition}
If Alice and Bob perform the above protocol with $v = \frac{32}{\eps}\left(t + \log(1/\eps) + 9\right)$, we have:
\begin{enumerate}
\item $Pr[\textrm{Bob succeeds}] \ge 2^{-t}$
\item $Pr[\textrm{Alice and Bob win } \ge (1-\frac{\eps}{256})n \textrm{ games }| \textrm{ Bob succeeds} ] \ge (1-\frac{\eps}{256})$.
\end{enumerate}
where  
$$Pr[\textrm{A\textsc{\&}B win } \ge (1-\frac{\eps}{256})n \textrm{ games }| \textrm{ Bob succeeds} ] = \Pr[ \#\{i : V(a_ib_i | x_i y_i) = 1\} \ge n(1 - \frac{\eps}{256}) | \textrm{ Bob succeeds}].$$
\end{proposition}
\begin{proof}
We first have:
\begin{align*}
\Pr[\textrm{Bob succeeds}] & = \Pr[\textrm{Alice and Bob win } G_i \ \forall i \in \{i_1,\dots,i_v\}] \\ & \ge \Pr[\textrm{Alice and Bob win } G_i \ \forall i \in [n] ] = 2^{-t}.
\end{align*}
As in the previous case, we have 
\begin{align*}
\Pr[\textrm{A\textsc{\&}B win } \le (1 -\frac{\eps}{32})n \textrm { games } | \textrm{ Bob succeeds}] & \le \frac{\Pr[\textrm{Bob succeeds} \mid \textrm{A\textsc{\&}B win }  \le (1 - \frac{\eps}{32})n \textrm { games}]}{\Pr[\textrm{Bob succeeds}]} \\
& \le \frac{\Pr[\textrm{Bob succeeds} \mid \textrm{A\textsc{\&}B win }  \le (1 - \frac{\eps}{256})n \textrm { games}]}{2^{-t}}.
\end{align*}

Moreover, we have 
\begin{align*}
 \Pr[ & \textrm{Bob succeeds} \mid \textrm{Alice and Bob win }  \le (1 - \frac{\eps}{256})n \textrm { games}] \\ =
& \ \Pr[a^C = \alpha^C \mid \textrm{Alice and Bob win }  \le (1 - \frac{\eps}{256})n \textrm { games}] + \Pr[a^C \neq \alpha^C] \cdot \Pr[h(a^C) = h(\alpha^C) \mid a^C \neq \alpha^C] \\ \le
& \ \Pr[a^C = \alpha^C \mid \textrm{Alice and Bob win }  \le (1 - \frac{\eps}{256})n \textrm { games}] + \Pr[h(a^C) = h(\alpha^C) \mid a^C \neq \alpha^C] \\ \le
& \ (1 - \frac{\eps}{256})^v + 2^{-2t}.
\end{align*}
Putting this together, we obtain
$$
\Pr[\textrm{Alice and Bob win } \le (1 - \frac{\eps}{256})n \textrm { games } | \textrm{ Bob succeeds}] \le \frac{(1 - \frac{\eps}{32})^v + 2^{-2t}}{2^{-t}} \le 1 - \frac{\eps}{256}
$$
since $v = \frac{32}{\eps}\left(t + \log(1/\eps) + 9\right)$.
\end{proof}
\paragraph{Getting parallel repetition}
\begin{theorem}\label{Theorem:ParallelRepetitionProjectionGames}
For any free projection game $G$, we have $\omega^*(G) \le (1 - \eps)^{\Omega(n)}$.
\end{theorem}
\begin{proof}
We proceed as in Theorem \ref{Theorem:ParallelRepetition}. If $t \ge \frac{n\eps}{2048}$, the statement holds. If $t \le \frac{n\eps}{2048}$, we can construct an state $\ket{\psi}$ and find an index $i$ such that 
\begin{enumerate}
\item $I(\X_i : \Bob)_{\ket{\psi}} + I(\Y_i : \Alice)_{\ket{\psi}} \ge \Omega(\eps)$
\item
$I(\X_i : \Bob)_{\ket{\psi}} + I(\Y_i : \Alice)_{\ket{\psi}} \le \frac{32t + 16|M_{A^C}|}{\overline{C}}$
\end{enumerate}
In this case, we have $|M_{A^C}| = 2t$ and $\overline{C} \ge \frac{6n}{7}$, which means that
$\frac{64t}{\overline{C}} \ge \Omega(\eps)$. This implies $t \ge \Omega(n/\eps)$ or equivalently $\omega^*(G^n) \le (1 - \eps)^{\Omega(n)}$.
\end{proof}

\bibliography{paper}
\bibliographystyle{alpha}

\appendix

\end{document}

%% file: new_macros.tex
\newtheorem{definition}{Definition}
\newtheorem{lemma}{Lemma}
\newtheorem{proposition}{Proposition}
\newtheorem{corollary}{Corollary}
\newtheorem{claim}{Claim}

\newtheorem{remk}{Remark}

\def\FullBox{\hbox{\vrule width 8pt height 8pt depth 0pt}}

\def\qed{\ifmmode\qquad\FullBox\else{\unskip\nobreak\hfil
\penalty50\hskip1em\null\nobreak\hfil\FullBox
\parfillskip=0pt\finalhyphendemerits=0\endgraf}\fi}

\newtheorem{fact}{Fact}

\def\qedsketch{\ifmmode\Box\else{\unskip\nobreak\hfil
\penalty50\hskip1em\null\nobreak\hfil$\Box$
\parfillskip=0pt\finalhyphendemerits=0\endgraf}\fi}

\newcommand{\etal}{{\it et~al.\ }}

\newcommand{\zo}{\{0,1\}}

\newcommand{\E}{\mathop{\mathbb E}\displaylimits}

\newcommand{\eps}{\varepsilon}

\ifnum\draft=1
\newcommand{\authnote}[2]{{ \bf [#1's Note: #2]}}
\else
\newcommand{\authnote}[2]{}
\fi

\newcommand{\COMMENT}[1]{}
\newcommand{\ket}[1]{|#1\rangle}
\newcommand{\bra}[1]{\langle#1|}
\newcommand{\ketbra}[2]{|#1\rangle\langle#2|}
\def\01{\{0,1\}}
\newcommand{\braket}[2]{\langle{#1}|{#2}\rangle} 
\def\01{\{0,1\}}

\newcommand{\triple}[3]{\langle{#1}|{#2}|{#3}\rangle}

\newcommand{\Tr}{\mbox{\rm Tr}}

\newcommand{\cadre}[1]
{
\begin{tabular}{|p{15.4cm}|}
\hline
#1 \\
\hline
\end{tabular}
}

\newcommand{\spa}[1]{\mathcal{#1}}

\newcommand{\X}{\mathcal{X}}
\newcommand{\Y}{\mathcal{Y}}